\documentclass{amsart}[12]
\usepackage{mathrsfs, amsmath,amsfonts,amssymb}
\usepackage{amsmath}
\usepackage{mathtools}
\usepackage{fullpage}

\theoremstyle{plain}

\newtheorem{definition}{Definition}

\newtheorem{proposition}{Proposition}

\numberwithin{equation}{section}

\author[N. Demni]{Nizar Demni}
\address{IRMAR, Universit\'e de Rennes 1\\ Campus de
Beaulieu\\ 35042 Rennes cedex\\ France}
\email{nizar.demni@univ-rennes1.fr}

\author[Z. Mouayn]{Zouhair Mouayn}
\address{Department of Mathematics\\ Faculty of Sciences and Technics (M'Ghila)\\ Sultan Moulay Slimane \\ PO. Box 523, B\'eni Mellal\\ Morocco}
\email{mouayn@fstbm.ac.ma}

\author[H. Yaqine]{Houda Yaqine}
\address{Department of Mathematics\\ Faculty of Sciences and Technics (M'Ghila)\\ Sultan Moulay Slimane \\ PO. Box 523, B\'eni Mellal\\ Morocco}
\email{yaqinehou@gmail.com}

\title{Berezin transforms attached to Landau levels on the complex projective space $\mathbb{CP}^n$}

\begin{document}

\maketitle

\begin{abstract}
In this paper, we construct coherent states for each generalized Bergman space on the $n$-dimensional complex projective space in order to apply a coherent states quantization method. Doing so allows to define the Berezin transform for these spaces. In particular, we provide a variational formula for this transform by means of the Fubini-Study Laplace operator which reduces when $n = 1$ and for the lowest spherical Landau levels to the well-known formula previously given by Berezin himself.
\end{abstract}

\section{Introduction} 
The Berezin transform was introduced and studied in \cite{Ber1}, \cite{Ber2}, for classical complex symmetric spaces and in \cite{Ber-Cob}, \cite{Eng}, \cite{HKZ} for the Bergman, Hardy and Bargman-Fock spaces. This is a quite relevant transform relating the covariant and the contravariant symbols of a bounded linear operator which is also useful in quantization theory and in the correspondence principle. In \cite{Ber1}, the author used the correspondence principle to express this transform in the sphere and in the Lobachevsky plane through their corresponding Laplace-Beltrami operators.
More generally, the Berezin transform stemming from systems of coherent states attached to generalized Bergman spaces on $\mathbb{C}^n, n \geq 1$ and the hyperbolic complex unit ball $\mathbb{B}^n$ were introduced in \cite{Mou1} and \cite{Gha-Mou} respectively. There, these spaces arise as eigenspaces of a Schr\"odinger operator with a uniform magnetic field corresponding to Euclidean and hyperbolic Landau levels and generalizations of Berezin formula were given (see also \cite{Pee} and for Berezin formulas on $\mathbb{B}^n$ and $\mathbb{C}^n$). Yet another expression for the Berezin transform on $\mathbb{B}^n$ in the magnetic setting were derived in \cite{Bou-Mou} using the Fourier-Helgason transform and involves the Wilson polynomials (\cite{AAR}.
In this paper, we are interested in the $n$-dimensional complex projective space $\mathbb{CP}^n, n \geq 1$ model endowed with its Fubini-Study metric for which we construct a system of coherent states. Actually, the latter are attached to eigenspaces of a Schr\"odinger operator with a uniform magnetic field in $\mathbb{CP}^n$ which we also refer to as generalized Bergman spaces. These coherent states are then used to define the Berezin transform in this magnetic setting, for which we provide a variational formula by means of the Fubini-Study Laplace operator. In particular, we recover the original formula of Berezin for the complex projective line which corresponds here to the lowest spherical Landau level. Note that these findings complete the aformentionned analyses of magnetic Berezin transforms on phase spaces modeled by rank-one complex symmetric spaces. Moreover, the existence of such variational formulas is ensured by the fact that the algebra of bi-invariant operators on these phase spaces are generated by their Laplace-Beltrami operators since the Berezin transforms commute with translation operators of the underlying Lie groups (\cite{Ber2}). At the physical level, it is worth noting that the variational formula is represented through the Fubini-Study Laplace operator which describes the motion of a free particle. Therefore, our strategy leading to this formula transfers the effect of the magnetic field at a higher spherical Landau level only the representing function. As a matter of fact, it sheds the light on the interplay between the geometry of the phase space on the one hand and the physical quantities on the other hand. This phenomenon has the same flavor as diamagnetic inequalities (see e.g. \cite{Sim}).
The paper is organized as follows. In the next section, we outline the formalism of coherent states quantization. In the third section, we recall from \cite{Haf-Int} the spectral theory of the magnetic Schr\"odinger operator on $\mathbb{CP}^n$ giving rise to the generalized Bergman spaces. In particular, we also exhibit the spectral function of the Fubini-Study Laplacian. In the fourth section, we attach a system of coherent states to each eigenspace and introduce the corresponding Berezin transform. The variational formula is then derived in the fifth section, whence we recover Berezin formula on the Riemann sphere.

\section{Coherent states quantization}
In this section, we recall the formalism of coherent states quantization and refer the reader to the standard monograph \cite{Gaz} (p.72-75).  
 Let $(X, \mu)$ be a measurable space and denote $L^{2}(X,\mu)$ the space of $\mu$-square integrable functions on $X$. Let $\mathcal{A} \subset L^{2}(X,d\mu )$ be a closed subspace (possibly infinite-dimensional) with an orthonormal basis $\left\{ \Phi _{j}\right\} _{j=0}^{\infty }$ and let $(\mathcal{H}, \langle \,\mid \,\rangle)$ be a infinite-dimensional separable Hilbert space equipped with an orthonormal basis $\left\{ \phi _{j}\right\} _{j=0}^{\infty }$. Then the coherent states $\left\{ \mid x>\right\} _{x\in X}$ in $\mathcal{H}$ are defined by
\begin{equation}\label{CohSta}
\mid x>:=\left( \mathcal{N}\left( x\right) \right) ^{-\frac{1}{2}} \sum_{j=0}^{+\infty }\Phi _{j}\left( x\right) \mid \phi _{j}>\quad 
\end{equation}
where
\begin{equation}\label{Norm}
\mathcal{N}\left( x\right) :=\sum_{j=0}^{+\infty }\Phi _{j}\left( x\right) \overline{\Phi _{j}\left( x\right) }
\end{equation}
is the normalization factor such that $\left\langle x\mid x\right\rangle _{\mathcal{H}}=1$. These states provide the following resolution of the identity operator:
\begin{equation}\label{Res}
\mathbf{1}_{\mathcal{H}}=\int\limits_{X}\mid x><x\mid \mathcal{N}\left(
x\right) d\mu \left( x\right) 
\end{equation}
where $\mid x><x\mid$ is the Dirac bra-ket notation for the rank-one operator $\varphi \mapsto \langle x \mid \varphi \rangle  x$. Note in passing that the choice of the Hilbert space $\mathcal{H}$ defines a quantization of the space $X$ by the coherent states defined by \eqref{CohSta} via the inclusion map $X\ni x\mapsto \mid x>\in \mathcal{H}$. In this respect, the property \eqref{Res} bridges between classical and quantum mechanics. In fact, the (Klauder-Berezin) coherent states quantization consists in associating to a classical observable (that is a function on $X$ with specific properties) the operator-valued integral: 
\begin{equation*}
A_f := \int_X \mid x> <x\mid f(x) \mathcal{N}(x)d\mu(x). 
\end{equation*}
The map $f \mapsto A_f$ is not one-to-one in general and for any given operator $A$ on $\mathcal{H}$, we call any function $f$ such that $A = A_f$ its upper or contravariant symbol. On the other hand, the mean value $\langle x \mid A \mid x \rangle$ of $A$ with respect to the coherent state $\mid x>$ is referred to as the lower or covariant symbol of $A$. Consequently, we can associate to a classical observable $f$ the expectation $\langle x \mid A_f \mid x \rangle$ and as such, we get the Berezin transform of $f$ defined by: 
\begin{equation*}
B[f](x) := \langle x \mid A_f \mid x \rangle, \quad x \in X. 
\end{equation*}

\section{Generalized Bergman spaces on $\mathbb{CP}^n$}
Below, we recall from \cite{Haf-Int} the spectral theory of the magnetic Schr\"odinger operator $\Delta_{\nu}$ in $\mathbb{CP}^n, n \geq 1$:
\begin{equation}\label{Oper}
\Delta_{\nu}:=4(1+|z|^2)\left(\sum_{i,j}^n(\delta_{ij}+z_i\bar{z_j})\frac{\partial^2}{\partial z_i \partial\bar{z_j}}+\nu\sum_{j=1}^n \left(z_j\frac{\partial}{\partial z_j}-\bar{z_j}\frac{\partial}{\partial \bar{z_j}}\right)-\nu^2\right)+4\nu^2,
\end{equation}
where $2\nu\in\mathbb{Z}^+$ (for sake of simplicity, we ommit the dependence of this operator on $n$). Namely, for $\lambda\in\mathbb{C}$, we consider the equation
\begin{equation}\label{3.1}
\Delta_{\nu}F(z)=\left(\lambda^2-n^2+4\nu^2\right)F(z),
\end{equation}
where $F$ is a bounded function on $\mathbb{C}^n$, and define
\begin{equation}\label{3.2}
A^{\nu}_{m}(\mathbb{CP}^n) := \{F:\mathbb{C}^n\to\mathbb{C},\ F\text{ is a bounded and }\Delta_{\nu}(z)=(\lambda^2-n^2+4\nu)F(z)\}.
\end{equation}
In order to recall the description of these eigenspaces, we need to fix some notations. For $p,q\in\mathbb{Z}_{+}$, let $H(p,q)$ denote the finite-dimensional space of all harmonic homogeneous polynomials in $\mathbb{C}^n$ which are 
of degree $p$ in $z$ and of degree $q$ in $\bar{z}$ and let $\mathscr{H}(p,q)$ be the space of spherical harmonics in $H(p,q)$. Then, the spaces $\mathscr{H}(p,q)$ are pairwise orthogonal in $L^2(S^{2n-1},dw)$ and the dimension of $\mathscr{H}(p,q)$ is given by 
\begin{equation}\label{3.3}
d(n,p,q)=\frac{(p+q+n-1)(p+n-2)!(q+n-2)!}{p!q!(n-1)!(n-2)!}.
\end{equation}
Pick any orthogonal basis $(h_{p,q})$ of $\mathscr{H}(p,q)$ and introduce the set 
\begin{equation}\label{3.4}
D_{\nu} := \{\lambda\in\mathbb{C},\ \frac{n\pm \lambda}{2}+\nu\in\mathbb{Z}_{-}\}\cup \{\lambda\in\mathbb{C},\ \frac{n\pm \lambda}{2}-\nu\in\mathbb{Z}_{-}\}.
\end{equation}
From \cite{Haf-Int}, p.147, the eigenspace  $A^{\nu}_{m}(\mathbb{CP}^n)=\{0\}$ if $\lambda\not\in D_{\nu}$, otherwise it is not trivial if and only if $\lambda$ has the form $\lambda=\pm (2(m+\nu)+n)$ for some $m\in\mathbb{Z}_{+}$. When this condition is fulfilled, any function $F(z)$ in $A^{\nu}_{m}(\mathbb{CP}^n)$ admits the expansion:
\begin{equation}\label{3.5}
F(z)=\left(1+|z|^2\right)^{-(m+\nu)}\sum_{\mathclap{\substack{0\leq p\leq m\\ 0\leq q\leq m+2\nu}}} \text{ }_2F_1\left(p-m,q-m-2\nu,n+p+q;-|z|^2\right)h_{p,q}(z,\bar{z}),
\end{equation}
where $_2F_1(a,b,c;x)$ is the Gauss hypergeometric function (\cite{AAR}) and $F$ satisfies the growth condition
\begin{equation}\label{3.6}
\lim_{r\to\infty}F(r\omega)=\sum_{0\leq p\leq m} (-1)^{m-p}\frac{\Gamma(m-p+1)\Gamma(n+2p+2\nu)}{\Gamma(m+n+p+2\nu)}h_{p,p+2\nu}(\omega,\bar{\omega}), \quad z=rw,\ r>0, \,\omega \in S^{2n-1}.
\end{equation}
Moreover, $A^{\nu}_{m}(\mathbb{CP}^n)$ is finite-dimensional and its dimension is given by
\begin{equation}\label{3.7}
(2m+n+2\nu)\frac{\Gamma(m+n)\Gamma(m+n+2\nu)}{n \Gamma^2(n)\Gamma(m+1)\Gamma(m+2\nu+1)}
\end{equation}
%and it is equipped with the Hermitian scalar product induced from $\left(L^2(\mathbb{C}^n),d\mu_n(z)\right)$
%\begin{equation}\label{3.8}
%<f,g>=\int\displaylimits_{\mathbb{C}^n}f(z)\overline{g(z)}d\mu_n(z)
%\end{equation}
Let 
\begin{equation}\label{1.8}
d\mu_n(w)=\frac{1}{(1+|w|^2)^{n+1}} dw
\end{equation}
where $dw$ is the Lebesgue measure on $\mathbb{C}^n$, then $A^{\nu}_{m}(\mathbb{CP}^n)$ admits a reproducing kernel given by
\begin{equation}\label{3.10}
K_m^{\nu}(z,w) := c^{\nu,n}_{m}\left[\cos d_{FS}(z,w)\right]^{2\nu}P_m^{(n-1,2\nu)}(\cos 2d_{FS}(z,w)),
\end{equation}
where $d_{FS}(z,w)$ is the Fubini-Study distance:
\begin{equation}\label{3.11}
\cos^2d_{FS}(z,w)=\frac{|1+ \langle z,w \rangle|^2}{(1+|z|^2)(1+|w|^2)}
\end{equation}
and 
\begin{equation}\label{Const}
c^{\nu,n}_{m} := \frac{(2m+2\nu+n)\Gamma(m+n+2\nu)}{\pi^n\Gamma(m+2\nu+1)}. 
\end{equation}

When $\nu=0$,  the operator $\Delta_{\nu}$ in \eqref{Oper} reduces to the Fubini-Study Laplacian:
\begin{equation}\label{2.5}
\Delta_0=4(1+|z|^2)\sum_{i,j}^n(\delta_{ij}+z_i\bar{z}_j)\frac{\partial^2}{\partial z_i\partial\bar{z}_j}
\end{equation}
which has a discrete spectral decomposition with eigenvalues $(-4k(k+n))_{k \geq 0}$. Besides, each eigenspace is finite-dimensional and has a orthonormal basis given by homogeneous spherical harmonics of degree zero. 
In particular, the kernel of the orthogonal projection from the $L^2(\mathbb{C}P^n, d_{FS})$ onto the $k$-th eigenspace reads:
\begin{equation}\label{SpecFun}
\psi_n(k;z,w):= \frac{(2k+n)\Gamma(n+k)}{\pi^n k!} \frac{P_k^{n-1, 0}(\cos 2d_{FS}(z,w))}{P_k^{n-1, 0}(1)}
\end{equation}
Thus, the spectral Theorem implies that for any suitable function $u$, the operator $u(-\Delta_{\textit{FS}})$ is an integral operator whose kernel is given by: 
\begin{equation}\label{Spect}
\sum_{k \geq 0} u(4k(k+n))\psi_n(k;z,w) .
\end{equation}

\section{Berezin transforms attached to spherical Landau levels}
Now, we specialize the definition \eqref{CohSta} of coherent states to the generalized Bergman space $A_m^{\nu}(\mathbb{CP}^n)$ by taking: 
\begin{itemize}
\item $X = \mathbb{CP}^n$ endowed with $d\mu_n$. 
\item $x \rightarrow z \in \mathbb{CP}^n$.
\item $\mathcal{A} \rightarrow A_m^{\nu}(\mathbb{CP}^n)$.    
\item A orthonormal basis $(\Phi_{p,q,j}^{\nu,m})_{p,q,j}$ of $A_m^{\nu}(\mathbb{CP}^n)$ where $1 \leq j \leq d(n,p,q)$ and $0 \leq q \leq m+2, 0 \leq p \leq m$. 
\item The Hilbert space $\mathscr{H}$ carrying the quantum states of some physical system and its basis $(\phi_{p,q,j})_{p,q,j}$ will be specified when needed.  
\end{itemize}
With these data, we define the following: 
\begin{definition}
For any $n \geq 1, 2\nu \in \mathbb{Z}_+$ and any $m \in \mathbb{Z}_+$, the system $\mid z; \nu, m >$ of coherent states attached to $A_m^{\nu}(\mathbb{CP}^n)$ is defined by 
\begin{equation*}
\mid z; \nu, m > := (\mathcal{N}^{\nu,m}(z))^{-1/2}\sum_{\mathclap{\substack{0\leq p\leq m\\ 0 \leq q \leq m+2\nu \\ 1 \leq j \leq d(n,p,q)}}} \overline{\Phi_{p,q,j}^{\nu,m}(z)} \mid \phi_{p,q,j} >
\end{equation*}
where $\mathcal{N}^{\nu,m}(z)$ is the normalizing factor given in \eqref{Norm}. 
\end{definition}
Observe that $\mathcal{N}^{\nu,m}(z)$ is the diagonal of the reproducing kernel $K_m^{\nu}(z,w)$ in \eqref{3.10}: 
\begin{equation*}
\mathcal{N}^{\nu,m}(z) = \sum_{\mathclap{\substack{0\leq p\leq m\\ 0 \leq q \leq m+2\nu \\ 1 \leq j \leq d(n,p,q)}}} \overline{\Phi_{p,q,j}^{\nu,m}(z)}\Phi_{p,q,j}^{\nu,m}(z) = K_m^{\nu}(z,z)
\end{equation*}
so that 
\begin{equation*}
\mathcal{N}^{\nu,m}(z) = \frac{(2m+2\nu+n)\Gamma(m+n+2\nu)}{\pi^n\Gamma(m+2\nu+1)}P_m^{(n-1,2\nu)}(1).
\end{equation*}
Using the special value (\cite{AAR}):
\begin{equation}\label{JacVal}
P_m^{(n-1,2\nu)}(1) = \frac{(n)_m}{m!},
\end{equation}
we readily get 
\begin{equation*}
\mathcal{N}^{\nu,m}(z) = \frac{(2m+2\nu+n)\Gamma(m+n+2\nu)}{\pi^{n}\Gamma(m+2\nu+1)}\frac{(n)_m}{m!}.
\end{equation*}
The states defined above satisfy the resolution of the identity 
\begin{equation*}
{\bf 1}_{\mathscr{H}} = \int_{\mathbb{CP}^n} \mid z; \nu, m > \, < z; \nu, m \mid \mathcal{N}^{\nu,m}(z)d\mu_n(z)
\end{equation*}
and allow for the quantization scheme described in section 2. As a matter of fact, we can define the Berezin transform of a classical observable $f$ in the usual way: we first associate to $f$ the operator-valued integral 
\begin{equation*}
A_f = \int_{\mathbb{CP}^n} \mid w; \nu, m > \, < w; \nu, m \mid f(w) \mathcal{N}^{\nu,m}(w)d\mu_n(w)
\end{equation*}
then take the expectation $\langle z; \nu, m  \mid A_f \mid  z; \nu, m  \rangle$ with respect to the coherent state $\mid z; \nu, m >$. 
\begin{definition}
The Berezin transform $B_m^{\nu}[f]$ of $f \in L^{\infty}(\mathbb{CP}^n, d\mu_n),$ attached to the generalized Bergman space $A_m^{\nu}(\mathbb{CP}^n)$ is defined by 
\begin{equation}\label{1.7}
\mathcal{B}_m[f](z)= \frac{c^{\nu,n}_{m}}{P^{(n-1,2\nu)}_{m}(1)} \int\displaylimits_{\mathbb{C}^n}\left(\cos^2d_{FS}(z,w)\right)^{2\nu}\left(P^{(n-1,2\nu)}_{m}(\cos 2d(z,w))\right)^2 f(w) d\mu_n(w).
\end{equation}
\end{definition}

Notice that the kernel of the Berezin transform $B_m^{\nu}$ depends only on the geodesic distance $d_{FS}$ therefore is a $SU(n+1,\mathbb{C})$-biinvariant function. It follows that $B_m^{\nu}$ commutes with the translation operators defined by group elements and in turn, 
it is a function of the Fubini-Study operator $\Delta_{FS}$ (\cite{Ber2}, p.353). In the subsequent section, we determine explicitly this function relying on the spectral theory of $\Delta_{FS}$. 
%\begin{equation}\label{4.1}
%\mathcal{B}_m[\varphi](z)=\int\displaylimits_{\mathbb{C}^n}B_m(z,w)\varphi(w)d\mu_n(w),
%\end{equation}
%where the kernel
%\begin{align}\label{4.2}
%B_m(z,w)&=\frac{|K_m^{\nu}(z,w)|^2}{K_m^{\nu}(z,z)}\nonumber \\&=\frac{2(m+2\nu+n)(m+2\nu+n-1)!}{\pi^n(m+2\nu)!}
%\left(\frac{|1+<z,w>|^2}{(1+|z|^2)(1+|w|^2)}\right)^{2\nu}\left(P_m^{(n-1,2\nu)}(\cos 2d(z,w))\right)^2
%\end{align}
%So that
%\begin{equation}\label{4.3}
%\mathcal{B}_m[\varphi](z)=c^{\nu,n}_{m}\int\displaylimits_{\mathbb{C}^n}\left(\frac{|1+<z,w>|^2}{(1+|z|^2)(1+|w|^2)}\right)^{2\nu}\left(P^{(n-1,2\nu)}_{m}(\cos 2d(z,w))\right)^2\varphi(w)d\mu_n(w)
%\end{equation}
%with
%\begin{equation}\label{4.4}
%c^{\nu,n}_{m}=\frac{2(m+2\nu+n)(m+2\nu+n-1)}{\pi^n(m+2\nu)!}=\text{Prefactor 1}
%\end{equation}

\section{Variational formula for the Berezin transform}
We seek a function $W = W_m^{\nu}$ depending on $m,\nu$ such that 
\begin{equation*}
B_m^{\nu} = W(-\Delta_{FS}). 
\end{equation*}
Appealing to \eqref{Spect}, the function $W$ should solve the equation
\begin{equation*}
 \frac{c^{\nu,n}_{m}}{P^{(n-1,2\nu)}_{m}(1)} \left(\cos^2d_{FS}(z,w)\right)^{2\nu}\left(P^{(n-1,2\nu)}_{m}(\cos 2d(z,w))\right)^2 = \sum_{k=0}^{\infty}W(\lambda_k)\psi_n(k,z,w)
\end{equation*}
where $\lambda_k := k(k+n)$. 
In order to get a solution to this equation, we need to expand the product of Jacobi polynomials in the left-hand side as a series of Jacobi polynomials $(P_k^{(n-1,0)})_{k\geq 0}$. To proceed, we make use of the following instance of formula (52) in \cite{Sri}, p.4467: for any $\alpha_1, \alpha, \beta > -1$, any $m, \mu \in \mathbb{Z}_+$ and any $t \in [0,1]$, one has:
\begin{align}\label{5.8}
t^{\mu}P_{m}^{(\alpha_1,\beta)}&(1-2t)P_{m}^{(\alpha_1,\beta)}(1-2t)\nonumber\\
&=(\alpha+1)_{\mu}\binom{\alpha_1+m}{m}^2
\,\,\sum_{k=0}^{+\infty}\frac{(\alpha+\beta+2k+1)(-\mu)_k}{(\alpha
+1)_k(\alpha+\beta+k+1)_{\mu+1}}P_k^{(\alpha,\beta)}(1-2t)\nonumber\\
&F^{2:2,2}_{2:1,1}\left[
\begin{matrix}
\mu+1,\ \alpha+\mu+1\ :\hfill -m\ ,\hfill \alpha_1+\beta+m+1\ ,\hfill -m\  
,\hfill \alpha_1+\beta+m+1\hfill \\ 
\mu-k+1\ ,\hfill \alpha+\beta+\mu+2+k\ :\hfill \alpha_1+1\ ,\hfill 
\alpha_1+1
\end{matrix}
\, \middle\vert \,
1,1
\right]
\end{align}
where 
\begin{equation*}
F^{2:2,2}_{2:1,1}\left[
\begin{matrix}
a_1, a_2 :\hfill b_1, b_2, b_3, b_4  \\ 
c_1, c_2 :\hfill d_1, d_2
\end{matrix}
\, \middle\vert \,
x,y
\right] = \sum_{s,l = 0}^{\infty} \frac{(a_1)_{l+s}(a_2)_{l+s}}{(c_1)_{l+s}(c_2)_{l+s}}\frac{(b_1)_l(b_2)_l(b_3)_s(b_4)_s}{(d_1)_l(d_2)_s} \frac{x^ly^s}{l!s!}
\end{equation*}
is the Kamp\'e de F\'eriet function (\cite{Man-Sri}). The issue of our computations is recorded in the following proposition: 
\begin{proposition}
The function $W = W_m^{\nu}$ can be chosen as: 
\begin{multline*}
W(\lambda) =  \gamma_{n,m,\nu}\frac{\Gamma(R_n(\lambda)+1)}{\Gamma(n+R_n(\lambda))\Gamma(2\nu-R_n(\lambda)+1)\Gamma(n+R_n(\lambda)+2\nu+1)} 
\\  \sum_{s=0}^m \frac{(-m)_s(2\nu+1)_s(2\nu+m+n)_s}{s! (2\nu-R_n(\lambda)+1)_{s}(n+2\nu+R_n(\lambda)+1)_{s}} 
{}_4F_3\left[
\begin{matrix}
-m,\ 2\nu+1+s, \, 2\nu+1+s ,2\nu+m+n, 
\\ 2\nu-R_n(\lambda)+1+s, n+2\nu+1+R_n(\lambda)+s, 2\nu+1
\end{matrix}; 1\right].
\end{multline*}
where 
\begin{equation*}
\gamma_{n,m,\nu} := (2m+2\nu+n) (m+2\nu+1)_{n-1} \frac{((2\nu+m)!)^2(n-1)!}{(n)_m m!} 
\end{equation*}
and  
\begin{equation*}
R_n(\lambda) := \frac{\sqrt{n^2+\lambda} - n}{2}, \quad \lambda \geq 0.
\end{equation*}
Consequently, the Berezin transform is given by $B_m^{\nu} = W_m^{\nu}(-\Delta_{FS})$. 
\end{proposition}

\begin{proof} 
Specializing (\ref{5.8}) with
\begin{equation*}
\mu= \alpha_1 = 2\nu, \, \, \beta = n-1, \, \alpha =0, \, t=\cos^2d(z,w), 
\end{equation*}
and using the symmetry relation $P_k^{\alpha,\beta}(-x) = (-1)^k P_k^{\alpha,\beta}(-x)$ (\cite{AAR}), we readily obtain
\begin{align*}
\left(\cos^2d_{FS}(z,w)\right)^{2\nu}\left(P^{(n-1,2\nu)}_{m}(\cos 2d(z,w))\right)^2  &=(1)_{2\nu} \binom{2\nu+m}{m}^2\sum_{k=0}^{\infty}\frac{(2k+n) (-2\nu)_k}{(n)_k(n+k)_{2\nu+1}}(-1)^k P_k^{(n-1,0)}(\cos 2d(z,w))  \nonumber 
\\ & \times F^{2:2,2}_{2:1,1}\left[
\begin{matrix}
2\nu+1,\ 2\nu+1 :\hfill -m\ ,\hfill 2\nu+m+n ,\hfill -m\  
,\hfill 2\nu+m+n\hfill \\ 
2\nu-k+1\ ,\hfill n+2\nu+k+1 :\hfill 2\nu+1\ ,\hfill 2\nu+1
\end{matrix}
\, \middle\vert \,
1,1
\right]\nonumber\\
&=\pi^{n}(1)_{2\nu} \binom{2\nu+m}{m}^2\sum_{k=0}^\infty \frac{(-1)^k(-2\nu)_kk!}{(n)_k(n+k)_{2\nu+1}(n+k-1)!} \psi_n(k;z,w)
%\\&  \nonumber\\
\\&\times F^{2:2,2}_{2:1,1}\left[
\begin{matrix}
2\nu+1,\ 2\nu+1 :\hfill -m\ ,\hfill 2\nu+m+n ,\hfill -m\  
,\hfill 2\nu+m+n\hfill \\ 
2\nu-k+1\ ,\hfill n+2\nu+k+1 :\hfill 2\nu+1\ ,\hfill 2\nu+1
\end{matrix}
\, \middle\vert \,
1,1
\right]
\\& =\pi^{n}(1)_{2\nu} \binom{2\nu+m}{m}^2\sum_{k=0}^\infty \frac{(-1)^k(-2\nu)_kk!}{(n)_k\Gamma(n+k+2\nu+1)} \psi_n(k;z,w)
 \nonumber\\
&\times F^{2:2,2}_{2:1,1}\left[
\begin{matrix}
2\nu+1,\ 2\nu+1 :\hfill -m\ ,\hfill 2\nu+m+n ,\hfill -m\  
,\hfill 2\nu+m+n\hfill \\ 
2\nu-k+1\ ,\hfill n+2\nu+k+1 :\hfill 2\nu+1\ ,\hfill 2\nu+1
\end{matrix}
\, \middle\vert \,
1,1
\right].
\end{align*}

Observe that the sum over $k$ terminates at $2\nu$ and that 
\begin{equation*}
(-2\nu)_k = (-1)^k \frac{(2\nu)!}{(2\nu-k)!}. 
\end{equation*}
Similarly, the Kamp\'e de Feriet series terminates at $m$ and using the relation 
\begin{equation*}
(a)_{l+s} = (a+s)_l(a)_s
\end{equation*}
satisfied by the Pochhammer symbol, we derive 
\begin{multline*}
F^{2:2,2}_{2:1,1}\left[
\begin{matrix}
2\nu+1,\ 2\nu+1 :\hfill -m\ ,\hfill 2\nu+m+n ,\hfill -m\  
,\hfill 2\nu+m+n\hfill \\ 
2\nu-k+1\ ,\hfill n+2\nu+k+1 :\hfill 2\nu+1\ ,\hfill 2\nu+1
\end{matrix}
\, \middle\vert \,
1,1
\right] = \sum_{l,s=0}^m \frac{(2\nu+1)_{l+s}(2\nu+1)_{l+s}}{(2\nu-k+1)_{l+s}(n+2\nu+k+1)_{l+s}} \\ \times \frac{(-m)_l(2\nu+m+n)_l(-m)_s(2\nu+m+n)_s}{(2\nu+1)_l(2\nu+1)_s} \frac{1}{l!s!}
= \sum_{s=0}^m \frac{(-m)_s(2\nu+1)_s(2\nu+m+n)_s}{s! (2\nu-k+1)_{s}(n+2\nu+k+1)_{s}} 
\\ \sum_{l=0}^m \frac{(-m)_l ((2\nu+1+s)_{l})^2  (2\nu+m+n)_l}{(2\nu-k+1+s)_{l}(n+2\nu+k+1+s)_{l}(2\nu+1)_l} \frac{1}{l!}
= \sum_{s=0}^m \frac{(-m)_s(2\nu+1)_s(2\nu+m+n)_s}{s! (2\nu-k+1)_{s}(n+2\nu+k+1)_{s}} \\
{}_4F_3\left[
\begin{matrix}
-m,\ 2\nu+1+s, ,\hfill 2\nu+1+s ,2\nu+m+n, 
\\ 2\nu-k+1+s, n+2\nu+1+k+s, 2\nu+1 
\end{matrix}; 1\right]
\end{multline*}
As a result, we should have 
\begin{multline*}
W(k(k+n)) =  \frac{((2\nu)!)^2\pi^n c^{\nu,n}_{m}}{P^{(n-1,2\nu)}_{m}(1)} \binom{2\nu+m}{m}^2 \frac{k!}{(n)_k\Gamma(2\nu-k+1)\Gamma(n+k+2\nu+1)} 
\\  \sum_{s=0}^m \frac{(-m)_s(2\nu+1)_s(2\nu+m+n)_s}{s! (2\nu-k+1)_{s}(n+2\nu+k+1)_{s}} 
{}_4F_3\left[
\begin{matrix}
-m,\ 2\nu+1+s, ,\hfill 2\nu+1+s ,2\nu+m+n, 
\\ 2\nu-k+1+s, n+2\nu+1+k+s, 2\nu+1
\end{matrix}; 1\right].
\end{multline*}
Keeping in mind \eqref{Const} and \eqref{JacVal}, the last expression simplifies as 
\begin{multline*}
W(k(k+n)) =  (2m+2\nu+n) (m+2\nu+1)_{n-1} \frac{((2\nu+m)!)^2}{m!} \frac{k!}{(n)_k(n)_m\Gamma(2\nu-k+1)\Gamma(n+k+2\nu+1)} 
\\  \sum_{s=0}^m \frac{(-m)_s(2\nu+1)_s(2\nu+m+n)_s}{s! (2\nu-k+1)_{s}(n+2\nu+k+1)_{s}} 
{}_4F_3\left[
\begin{matrix}
-m,\ 2\nu+1+s, ,\hfill 2\nu+1+s ,2\nu+m+n, 
\\ 2\nu-k+1+s, n+2\nu+1+k+s, 2\nu+1
\end{matrix}; 1\right].
\end{multline*}
Solving the equation $k(k+n) = \lambda$ in the variable $k \geq 0$ for $\lambda \geq 0$, we are done. 
\end{proof}

When $m=0, n=1$, we retrieve Berezin formula in the case of the Riemann sphere. Indeed, for these parameters, the function $W_m^{\nu}$ reduces to 
\begin{equation*}
W(k(k+1)) =  \gamma_{1,0,\nu} \frac{1}{\Gamma(2\nu-k+1)\Gamma(k+2\nu+2)}. 
\end{equation*}
Now, recall the Weierstrass product for the Gamma function: 
\begin{equation*}
\frac{1}{\Gamma(s+1)} = e^{\gamma s} \prod_{p \geq 1} \left(1+\frac{s}{p}\right)e^{-s/p}
\end{equation*}
where $\gamma$ is the Euler constant. 
It follows that 
\begin{equation*}
W(k(k+n)) =  \gamma_{1,0,\nu}e^{\gamma(4\nu+3)}\prod_{p \geq 1}\left(1+\frac{2\nu-k}{p}\right)\left(1+\frac{2\nu+1+k}{p}\right)e^{-(4\nu+3)/p}. 
\end{equation*}
Writing 
\begin{equation*}
\left(1+\frac{2\nu-k}{p}\right)\left(1+\frac{2\nu+1+k}{p}\right) = \left(1+\frac{2\nu}{p}\right)\left(1+\frac{2\nu+1}{p}\right) \left(1-\frac{k}{p+2\nu}\right)\left(1+\frac{k}{p+2\nu+1}\right)
\end{equation*}
and using again he Weierstrass product 
\begin{equation*}
e^{\gamma(4\nu+3)}\prod_{p \geq 1}\left(1+\frac{2\nu}{p}\right)\left(1+\frac{2\nu+1}{p}\right)e^{-(4\nu+3)/p} = \frac{1}{\Gamma(2\nu+1)\Gamma(2\nu+2)},
\end{equation*}
we get 
\begin{equation*}
W(k(k+1)) =   \frac{\gamma_{1,0,\nu}}{\Gamma(2\nu+1)\Gamma(2\nu+2)}\prod_{p \geq 1} \left(1- \frac{k(k+1)}{(p+2\nu)(p+2\nu+1)}\right). 
\end{equation*}
As a matter of fact, $W$ can be chosen as 
\begin{equation*}
W(\lambda) =   \frac{((2\nu)!)^2(2\nu+1)}{\Gamma(2\nu+1)\Gamma(2\nu+2)}\prod_{p \geq 1} \left(1- \frac{\lambda}{(p+2\nu)(p+2\nu+1)}\right), \quad \lambda \geq 0 
\end{equation*}
so that (we use the identity $\Gamma(s+1) = s\Gamma(s)$)
\begin{equation*}
B_0^{\nu} = W(-\Delta_{FS}) = \prod_{p \geq 1} \left(1 + \frac{\Delta_{FS}}{(p+2\nu)(p+2\nu+1)}\right).
\end{equation*}
Identifying $2\nu$ to $1/h$ in the notation of \cite{Ber1} (see eq. (5.9) p. 171), we get Berezin's formula: 
\begin{equation*}
B_0^{\nu}  = \prod_{p \geq 1} \left(1 + h^2\frac{\Delta_{FS}}{(1+ph)(1+(p+1)h)}\right).
\end{equation*}
A similar formula holds for $m=0$ and general $n \geq 1$.

\end{document}